\documentclass[runningheads, envcountsame, a4paper]{llncs}
\usepackage{amsmath}
\usepackage{amssymb}
\usepackage{theorem}
\usepackage{graphicx}
\usepackage{tikz}
\usepackage{tabularx}
\usepackage[noend,ruled,linesnumbered]{algorithm2e}

\def\pp{\mathinner{\ldotp\ldotp}}
\def\sa#1{\mbox{\tt #1}}

 \def\PV{\mathcal{P}}
 
 \def\PVW{\mathcal{P}_{\s{w}}}
 \def\QV{\mathcal{Q}} 
\def\undefined{|w|}

\def\pp{\ldotp\ldotp}
\renewcommand{\epsilon}{\varepsilon}

\def\s#1{\mbox{$#1$}}

\begin{document}

\title{Online Computation of Abelian Runs\thanks{Final version to be published in the Proceedings of LATA 2015.}}
\toctitle{Online Computation of Abelian Runs}

\author{Gabriele Fici\inst{1} \and Thierry Lecroq\inst{2} \and Arnaud Lefebvre\inst{2} \and \'Elise Prieur-Gaston\inst{2}}
\tocauthor{Gabriele~Fici}
\tocauthor{Thierry~Lecroq}
\tocauthor{Arnaud~Lefebvre}
\tocauthor{\'Elise~Prieur-Gaston}

\authorrunning{G. Fici, T. Lecroq, A. Lefebvre, and \'E. Prieur-Gaston}

\institute{
Dipartimento di Matematica e Informatica, Universit\`a di Palermo, Italy\\ \email{Gabriele.Fici@unipa.it}  \and
Normandie Universit\'e, LITIS EA4108, NormaStic CNRS FR 3638, IRIB, Universit\'e de Rouen, 76821 Mont-Saint-Aignan Cedex, France \\ \email{$\{$Thierry.Lecroq,Arnaud.Lefebvre,Elise.Prieur$\}$@univ-rouen.fr}}

\maketitle

\setcounter{footnote}{0} 

\begin{abstract}
Given a word $w$ and a Parikh vector $\PV$, an abelian run of period $\PV$ in $w$ is a maximal occurrence of a substring of $w$ having abelian period $\PV$. We give an algorithm that finds all the abelian runs of period $\PV$ in a word of length $n$ in time $O(n\times |\PV|)$ and space $O(\sigma+|\PV|)$.
\keywords{Combinatorics on Words, Text Algorithms, Abelian Period, Abelian Run}
\end{abstract}

\section{Introduction}

Computing maximal (non-extendable) repetitions in a string is a classical topic in the area of string algorithms (see for example \cite{Smy2013} and references therein). Detecting maximal repetitions of substrings, also called \emph{runs}, gives information on the repetitive regions of a string, and is used in many applications, for example in the analysis of genomic sequences.

Kolpakov and Kucherov \cite{KK99} gave a linear time algorithm for computing all the runs in a word and conjectured that any word of length $n$ contains less than $n$ runs. Bannai et al.~\cite{B14} recently proved this conjecture using the notion of Lyndon root of a run.

Here we deal with a generalization of this problem to the commutative setting. Recall that an abelian power is a concatenation of two or more words that have the same Parikh vector, i.e., that have the same number of occurrences of each letter of the alphabet. For example, $aababa$ is an abelian square, since $aab$ and $aba$ both have $2$ $a$'s and $1$ $b$. When an abelian power occurs within a string, one can search for its ``maximal'' occurrence by extending it to the left and to the right character by character without violating the condition on the number of occurrences of each letter. Following the approach of Constantinescu and Ilie~\cite{CI2006}, we say that a Parikh vector $\PV$ is an abelian period for a word $w$ over a finite ordered  alphabet $\Sigma=\{a_{1},a_{2},\ldots , a_{\sigma}\}$ if $w$ can be written as $w=u_0u_1 \cdots u_{k-1}u_k$ for some $k>2$ where for $0<i<k$ all
 the $u_i$'s have the same Parikh vector $\PV$ and the Parikh vectors of $u_0$ and $u_k$ are contained in $\PV$. Note that the factorization above is not necessarily unique. For example, $a\cdot bba\cdot bba \cdot \epsilon$ and $\epsilon \cdot abb\cdot abb \cdot a$ ($\epsilon$ denotes the empty word) are two factorizations of the word $abbabba$ both corresponding to the abelian period $(1,2)$. Moreover, the same word can have different abelian periods.
 
 In this paper we define an \emph{abelian run} of period $\PV$ in a word $w$ as an occurrence of a substring $v$ of $w$ such that $v$ has abelian period $\PV$ and this occurrence cannot be extended to the left nor to the right by one letter into a substring having the same abelian period $\PV$. 
 
For example, let $w=ababaaa$. Then the prefix $ab\cdot ab\cdot a=w[1\pp 5]$ has abelian period $(1,1)$ but it is not an abelian run since the prefix $a\cdot ba \cdot ba\cdot a=w[1\pp 6]$ has also abelian period $(1,1)$. This latter, instead, is an abelian run of period $(1,1)$ in $w$.
 
 Looking for abelian runs in a string can be useful to detect those regions in a string in which there is some kind of non-exact repetitiveness, for example  regions in which there are several consecutive occurrences of a substring or its reverse.
 
Matsuda et al.~\cite{matsuda2014computing} recently presented an offline algorithm for computing all abelian runs of a word of length $n$ in $O(n^2)$ time.
 Notice that, however, the definition of abelian run in~\cite{matsuda2014computing} is slightly different from the one we consider here. We will comment on this in Section \ref{sect-prev}.
 
We present an online algorithm that, given a word $w$ of length $n$ over an alphabet of cardinality $\sigma$, and a Parikh vector $\PV$, returns all the abelian runs of period $\PV$ in $w$ in time  $O(n\times |\PV|)$ and space $O(\sigma+|\PV|)$.

\section{Definitions and Notation}

Let $\Sigma=\{a_{1},a_{2},\ldots,a_{\sigma}\}$ be a finite ordered
 alphabet of cardinality $\sigma$ and let $\Sigma^*$ be the set of finite words
 over $\Sigma$. 
We let $|\s{w}|$ denote the length of the word $\s{w}$.
Given a word $\s{w}=\s{w}[0\pp n-1]$ of length $n>0$, we write $\s{w}[i]$ for the $(i+1)$-th symbol of $\s{w}$
 and, for $0\leqslant i \leqslant j< n$, we write $\s{w}[i\pp j]$ for the substring of $\s{w}$
 from the $(i+1)$-th symbol to the $(j+1)$-th symbol, both included.
We let $|\s{w}|_a$ denote the number of occurrences of the symbol
 $a\in\Sigma$ in the word $\s{w}$. 

The \emph{Parikh vector} of $\s{w}$, denoted by $\PVW$,
 counts the
 occurrences of each letter of $\Sigma$ in $\s{w}$, that is, 
 $\PVW=(|\s{w}|_{a_{1}},\ldots,|\s{w}|_{a_{\sigma}})$.
Notice that two words have the same Parikh vector if and only if
 one word is a permutation (i.e., an anagram) of the other.

Given the Parikh vector $\PVW$ of a word $\s{w}$, we let $\PVW [i]$ denote its
 $i$-th component and $|\PVW|$ its norm, defined as the sum of its components.
Thus, for $\s{w}\in\Sigma^*$ and $1\leqslant i\leqslant\sigma$, we have
 $\PVW [i]=|\s{w}|_{a_i}$ and $|\PVW|=\sum_{i=1}^{\sigma}\PVW[i]=|\s{w}|$.

Finally, given two Parikh vectors $\PV,\QV$, we write $\PV\subset \QV$ if
 $\PV[i]\leqslant \QV[i]$
 for every $1\leqslant i\leqslant \sigma$ and $|\PV|<|\QV|$. 
 
\begin{definition}[Abelian period \cite{CI2006}]
\label{def-ap}
A Parikh vector $\PV$ is an abelian period for a word $\s{w}$  if
 $\s{w}=\s{u}_0\s{u}_1 \cdots \s{u}_{k-1}\s{u}_{k}$, for some $k>2$, where
 $\PV_{\s{u}_{0}}\subset \PV_{\s{u}_{1}}=\cdots =\PV_{\s{u}_{k-1}}\supset \PV_{\s{u}_{k}}$,
 and $\PV_{\s{u}_{1}}=\PV$.

\end{definition}

Note that since the Parikh vector of $\s{u}_0$ and $\s{u}_k$ cannot be included
 in $\PV$ it implies that $|\s{u}_0|, |\s{u}_k|<|\PV|$.
We call $\s{u}_0$ and $\s{u}_k$ respectively the \emph{head} and the
 \emph{tail}  of the abelian period.
Note that in~\cite{CI2006} the abelian period is characterized by $|\s{u}_0|$ and $|\PV|$ thus we will sometimes use the notation $(h,p)$ for an abelian period of norm $p$ and head length $h$ of a word $w$.  Notice that the length $t$ of the tail is uniquely determined by $h$, $p$ and $n=|\s{w}|$, namely $t=(n-h) \bmod p$. 

\begin{definition}[Abelian repetition]
\label{def-arep}
A substring $w[i\pp j]$ is an abelian repetition with period length $p$ if $i-j+1$ is a multiple of $p$,
 $i-j+1 \ge 2p$
 and there exists a Parikh vector $\PV$ of norm $p$ such that $\PV_{\s{w}[i+kp\pp i+(k+1)p-1]}=\PV$ for every $0\le k \le p/(i-j+1)$.
\end{definition}

An abelian repetition $w[i\pp j]$ with period length $p$ such that
 $i-j+1 = 2p$ is called an abelian square.
An abelian repetition $w[i\pp j]$ of period length $p$ of a string $w$ is maximal if:
\begin{enumerate}
\item $\PV_{\s{w}[i−p\pp i−1]} \ne \PV_{\s{w}[i\pp i+p−1]}$ or $i − p < 0$;
\item $\PV_{\s{w}[j−p+1\pp j]} \ne \PV_{\s{w}[j+1\pp j+p]}$ or $j + p \ge n$.
\end{enumerate}

We now give the definition of an abelian run. Let $\s{v}=\s{w}[b\pp e]$, $0\le b \le e \le |\s{w}|-1$, be an occurrence of a substring in $\s{w}$ and suppose that $\s{v}$ has an abelian period $\PV$, with head length $h$ and tail length $t$. Then we denote this occurrence by the tuple $(b,h,t,e)$.

\begin{definition}
Let $\s{w}$ be a word. An occurrence $(b,h,t,e)$ of a substring of $\s{w}$ starting at position $b$, ending at position $e$, and having abelian period $\PV$ with head length $h$ and tail length $t$ is called  \textbf{left-maximal} (resp.~\textbf{right maximal}) if there does not exist an occurrence of a substring $(b-1,h',t',e)$ (resp.~$(b,h',t',e+1)$) with the same abelian period $\PV$. An occurrence $(b,h,t,e)$ is called \textbf{maximal} if it is both left-maximal and right-maximal.
\end{definition}

This definition leads to the one of abelian run.

\begin{definition}
An \textbf{abelian run} is a maximal occurrence $(b,h,t,e)$ of a substring with abelian period $\PV$ of norm $p$ such that  $(e-b-h-t+1) \ge 2 p$ (see \figurename~\ref{figu-def}).
\end{definition}

\begin{figure}
\begin{tikzpicture}[scale=0.8]

\draw (0,0) rectangle (2,.5);
\draw (2,0) rectangle (4,.5);
\draw (4,0) rectangle (7,.5);
\draw (7,0) rectangle (10,.5);
\draw (10,0) rectangle (13,.5);
\draw (13,0) rectangle (14,.5);
\draw (14,0) rectangle (15,.5);

\node at (2.25,.75) {$b$};
\node at (13.75,.75) {$e$};

\draw[<->] (2,-.5) -- (4,-.5) ;
\node at (3,-1) {$h$};

\draw[<->] (13,-.5) -- (14,-.5) ;
\node at (13.5,-1) {$t$};

\draw[<->] (4,-.5) -- (7,-.5) ;
\node at (5.5,-1) {$\PV$};

\draw[<->] (7,-.5) -- (10,-.5) ;
\node at (8.5,-1) {$\PV$};

\draw[<->] (10,-.5) -- (13,-.5) ;
\node at (11.5,-1) {$\PV$};

\draw[<->] (4,1) -- (13,1) ;
\node at (8.5,1.5) {$e-b-h-t+1$};

\end{tikzpicture} 

\caption{\label{figu-def}
The tuple $(b,h,t,e)$ denotes an occurrence of a substring starting at position $b$, ending at position $e$, and having abelian period $\PV$ with head length $h$ and tail length $t$.
}
\end{figure}
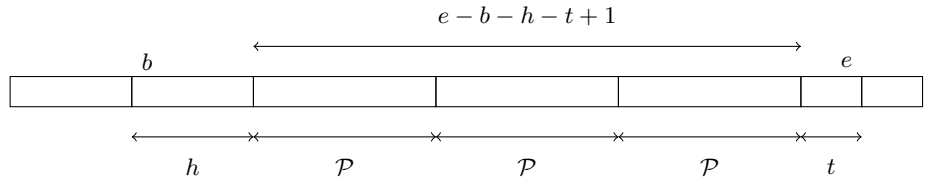

The next result limits the number of abelian runs starting at each position in a word.

\begin{lemma}\label{lem-run}
Let $\s{w}$ be a word. Given a Parikh vector $\PV$, there is at most one abelian run with abelian period $\PV$ starting at each position of $\s{w}$.
\end{lemma}

\begin{proof}
If two abelian runs start at the same position, the one
 with the shortest head cannot be maximal.
 \qed
\end{proof}

\begin{corollary}\label{coro-run}
Let $\s{w}$ be a word. Given a Parikh vector $\PV$, for every position $i$ in $\s{w}$ there are at most $|\PV|$ abelian runs with period $\PV$ overlapping at $i$.
\end{corollary}

The next lemma shows that a left-maximal abelian substring at the right of another left-maximal
 abelian substring starting at position $i$ in a word $w$ cannot begin at a position smaller than $i$.
 
\begin{lemma}\label{lemma-onecandidate}
If $(b_1,h_1,0,e_1)$ and $(b_2,h_2,0,e_2)$ are two left-maximal occurrences of substrings with the same abelian period $\PV$ of a word $\s{v}$
 such that $e_1 < e_2$ and
 $b_1 > e_1-2\times |\PV|+1$ and $b_2 > e_2-2\times |\PV|+1$,
 then $b_1 \le b_2$.
\end{lemma}

\begin{proof}
If $b_2<b_1$ then since $e_2 > e_1$, $\s{w}[b_1\pp b_1+h_1-1]$ is a substring of $\s{w}[b_2\pp b_2+h_2-1]$.
Thus $\PV_{\s{w}[b_1\pp b_1+h_1-1]} \subset \PV_{\s{w}[b_2\pp b_2+h_2-1]} \subset \PV$ which implies
 that $\PV_{\s{w}[b_1-1\pp b_1+h_1-1]} \subset \PV$ meaning that $(b_1,h_1,0,e_1)$ is not left-maximal:
 a contradiction.
 \qed
 \end{proof}

We recall the following proposition, which shows that if we can extend the abelian period with the
 longest tail of a word $w$ when adding a symbol $a$, then we can extend all
 the other abelian periods with shorter tail.

\begin{proposition}[\cite{DAM}]
\label{prop-heap}
Suppose that a word $w$ has $s$ abelian periods $(h_1,p_1) < (h_2,p_2) < \cdots <(h_s,p_s)$
 such that $(|w|-h_i) \mod p_i=t>0$ for every $1\le i \le s$.
If for a letter $a\in\Sigma$, $(h_1,p_1)$ is an abelian period of $wa$,
 then $(h_2,p_2),\ldots,(h_s,p_s)$ are also abelian periods of $wa$. 
\end{proposition}

We want to give an algorithm that, given a string $\s{w}$ and a Parikh vector
 $\PV$, returns all the abelian runs of $\s{w}$ having abelian period $\PV$.

\section{Previous Work}\label{sect-prev}

In~\cite{matsuda2014computing}, the authors presented an algorithm that computes all the abelian runs
 of a string $w$ of length $n$ in $O(n^2)$ time and space complexity.
They consider that a substring $w[i-h\pp j+t]$ is an abelian run if
  $w[i\pp j]$ is a maximal abelian repetition with period length $p$ and
  $h,t \ge 0$ are the largest integers satisfying
  $\PV_{w[i-h\pp i-1]}\subset \PV_{w[i\pp i+p-1]}$ and
  $\PV_{w[j+1\pp j+t]}\subset \PV_{w[i\pp i+p-1]}$.
Their algorithm works as follows.
First, it computes all the abelian squares using the algorithm of~\cite{Cummings_weakrepetitions}.
For each $0 \le i \le n-1$, it computes a set $L_i$ of integers such that
\[
L_i=\{j \mid \PV_{w[i-j\pp i]} = \PV_{w[i+1\pp i+j+1]}, 0 \le j \le \min\{i+1,n-i\}\}.
\] 
The $L_i$'s are stored in a two-dimensional boolean array $L$ of size $\lfloor n/2 \rfloor \times (n-1)$:
 $L[j,i]=1$ if $j\in L_i$ and $L[j,i]=0$ otherwise.
An example of array $L$ is given in Figure~\ref{figu-jap}.
All entries in $L$ are initially unmarked.
Then, for each $1\le j \le \lfloor n/2 \rfloor \times$ all maximal abelian repetitions of period length $j$ are computed in $O(n)$.
The $j$-th row of $L$ is scanned in increasing order of the column index.
When an unmarked entry $L[j,i]=1$ is found then the largest non-negative integer $k$
 such that $L[j,i+pj+1]=1$, for $1\le p \le k$, is computed.
This gives a maximal abelian repetition with period length $j$ starting at position $i-j+1$
 and ending at position $i+(k+1)j$.
Meanwhile all entries $L[j,i+pj+1]$, for $-1\le p \le k$, are marked.
Thus all abelian repetitions are computed in $O(n^2)$ time.
It remains to compute the length of their heads and tails.
This cannot be done naively otherwise it would lead to a 
 $O(n^3)$ time complexity overall.
Instead, for each $0\le i \le n-1$, let $T_i$ be the set of positive integers such that
 for each $j\in T_i$ there exists a maximal abelian repetition of period $j$ and starting
  at position $i-j+1$.
Elements of $T_i$ are processed in increasing order.
Let $j_k$ denote the $k$-th smallest element of $T_i$.
Let $h_k$ denote the length of the head of the abelian run computed from the abelian repetition of period $j_k$.
Then $h_{k}$ can be computed from $h_{k-1}$, $j_{k-1}$ and $j_{k}$
 as follows.
Two cases can arise:
\begin{enumerate}
\item
If $k=0$ or $j_{k-1}+h_{k-1}\le j_k$, then $h_k$ can be computed by comparing the
 Parikh vector $\PV_{w[i-j_k-p\pp i-j_k]}$ for increasing values of $\PV$ from $0$ up to
 $h_k+1$, with the Parikh vector $\PV_{w[i-j_k+1\pp i]}$.
\item
If $j_{k-1}+h_{k-1} > j_k$, then\\
 $\PV_{w[i-j_{k-1}-h_k\pp i-j_k]}$ can be computed from
 $\PV_{w[i-j_{k-1}-h_{k-1}+1\pp i-j_{k-1}]}$. Then, $h_k$ is computed by comparing the Parikh vector
 $\PV_{w[i-j_{k-1}-h_{k-1}+1-p\pp i-j_k]}$ for increasing values of $p$ from $0$ up to
 $h_k+j_k-h_{k-1}-j_{k-1}+1$.
\end{enumerate}
This can be done in $O(n)$ time.
The lengths of the tails can be computed similarly.
Overall, all the runs can be computed in time and space $O(n^2)$.

\begin{figure}
\begin{center}

\begin{tabularx}{.7\textwidth}{|X||X|X|X|X|X|X|X|X|X|X|X|X|}
\hline
&\sa{a}&\sa{b}&\sa{a}&\sa{a}&\sa{b}&\sa{a}&\sa{b}&\sa{a}&\sa{a}&\sa{b}&\sa{b}&\sa{b}\\
\hline
&0&1&2&3&4&5&6&7&8&9&10&11\\
\hline\hline
1&0&0&1&0&0&0&0&1&0&1&1&0\\
\hline
2&0&0&0&0&1&1&0&1&0&0&0&0\\
\hline
3&0&0&1&0&1&1&0&0&0&0&0&0\\
\hline
4&0&0&0&0&0&0&1&0&0&0&0&0\\
\hline
5&0&0&0&0&1&0&0&0&0&0&0&0\\
\hline
6&0&0&0&0&0&0&0&0&0&0&0&0\\
\hline
\end{tabularx}
\end{center}
\caption{\label{figu-jap}
An example of array $L$ for $w=\sa{abaababaabbb}$.
$L_{4,6}=1$ which means that $\PV_{w[3\pp 6]}=\PV_{w[7\pp 10]}$.
}
\end{figure}

This previous method works offline: it needs to know the whole string before
 reporting any abelian run.
We will now give what we call an online method meaning that we will be able
 to report the abelian runs ending at position $i-1$ of a string $w$
 when processing position $i$.
 However, this method is restricted to a given Parikh vector.
 
\section{A Method for Computing Abelian Runs of a Word with a Given Parikh Vector}

\subsection{Algorithm}

Positions of $w$ are processed in increasing order.
Assume that when processing position $i$ we know all the, at most $|\PV|$, abelian substrings ending at position $i-1$.
At each position $i$ we checked if $\PV_{w[i-|\PV|+1\pp i]} = \PV$ then all abelian substrings ending at position $i-1$
 can be extended and thus become abelian substrings ending at position $i$.
Otherwise, if $\PV_{w[i-|\PV|+1\pp i]} \ne \PV$ then abelian substrings ending at positions $i-1$ are processed
 in decreasing order of tail length.
When an abelian substring cannot be extended it is considered as an abelian run candidate.
As soon as an abelian substring ending at position $i-1$ can be extended then all the others (with smaller tail length)
 can be extended: they all become abelian substrings ending at position $i$.
At most one candidate (with the smallest starting position) can be output at each position.

\subsection{Implementation}

The algorithm \textsc{Runs}$(\PV,w)$ given below computes all the
 abelian runs with Parikh vector $\PV$ in the word $w$.
 It uses:
 \begin{itemize}
 \item function \textsc{Find}$(\PV,w)$, which returns the ending position of the first occurrence
  of Parikh vector $\PV$ in $w$ or $|w|+1$ if such an occurrence does not exist;
 \item function \textsc{FindHead}$(w,i,\PV)$, which returns the leftmost position $j<i$ such that
  $\PV_{w[j\pp i-1]} \subset \PV$ or $i$ is such a substring does not exist;
 \item function \textsc{Min}$(B)$ that returns the smallest element of the integer array $B$.
 \end{itemize}

Positions of $w$ are processed in increasing order (Lines~\ref{algo-loop1a}--\ref{algo-loop1b}).
We will now describe the situation when processing position $i$ of $w$:
\begin{itemize}
\item
array $B$ stores the starting positions of abelian substrings ending at position $i-1$ for the different $|\PV|$
 tail lengths ($B$ is considered as a circular array);
\item
$t_0$ is the index in $B$ of the possible abelian substring with a tail of length $0$ ending at position $i$.
\end{itemize}

All the values of the array $B$ are initially set to $|w|$.
Then, when processing position $i$ of $w$, for $0\le k < |\PV|$ and $k\ne t_0$,
 if $B[k]=b<|w|$ then
 $w[b\pp i-1]$ is an abelian substring with Parikh vector $\PV$ with tail length
 $((t_0-k+|\PV|)\bmod |\PV|)-1$.
Otherwise, if $B[k]=|w|$ then it means that there is no abelian substring in $w$
 ending at position $i-1$ with tail length $((t_0-k+|\PV|)\bmod |\PV|)-1$.

The algorithm \textsc{Runs}$(\PV,w)$ uses two other functions:
 \begin{itemize}
 \item function \textsc{GetTail}$(tail,t_0,p)$, which returns $(t_0-tail+p)\bmod p$ which is the length of the tail
  for the abelian substring ending at position $i-1$ and starting in $B[tail]$;
 \item function \textsc{GetRun}$(B,tail,t_0,e,p)$, which returns the abelian substring\\ $(B[tail],h,t,e)$.
 \end{itemize}

If $\PV_{w[i-|\PV|+1\pp i]} = \PV$ (Line~\ref{algo-cas1a}) then all abelian substrings ending at position
 $i-1$ can be extended (see \figurename~\ref{figu-cas1}).
Either this occurrence does not extend a previous occurrence at position $i-|\PV|$
 (Line~\ref{algo-cas1b}): the starting position
 has to be stored in $B$ (Line~\ref{algo-cas1c}) or
 this occurrence extends a previous occurrence at position $i-|\PV|$ and the starting position
 is already stored in the array $B$.

\begin{figure}

\begin{tikzpicture}[scale=0.8]

\draw (0,0) rectangle (2,.5);
\draw (2,0) rectangle (4,.5);
\draw (4,0) rectangle (7,.5);
\draw (7,0) rectangle (10,.5);
\draw (10,0) rectangle (12.5,.5);
\draw (12.5,0) rectangle (13,.5);
\draw (13,0) rectangle (15,.5);

\node at (12.75,.75) {$i$};

\draw[<->] (4,-.5) -- (7,-.5) ;
\node at (5.5,-1) {$|\PV|$};

\draw[<->] (7,-.5) -- (10,-.5) ;
\node at (8.5,-1) {$|\PV|$};

\draw[<->] (10,-.5) -- (13,-.5) ;
\node at (11.5,-1) {$|\PV|$};

\draw (11,-1.5) rectangle (13,-2);
\node at (12.75,-2.25) {$\cdots$};
\draw (12,-2.5) rectangle (13,-3);

\end{tikzpicture} 
\caption{\label{figu-cas1}
If $\PV_{w[i-|\PV|+1\pp i]}=\PV$ then $\PV_{w[j\pp i]}\subset\PV$ for
 $i-|\PV|+1 < j < i$.
}
\end{figure}

If $\PV_{w[i-|\PV|+1\pp i]} \ne \PV$ (Lines~\ref{algo-cas2a}-\ref{algo-loop1b})   
 then abelian substrings ending at position $i-1$ are processed
 in decreasing order of tail length.
To do that, the circular array $B$ is processed in increasing order
  of index starting from $t_0$ (Lines~\ref{algo-cas2b}-\ref{algo-cas2i}).
  
Let $tail$ be the current index in array $B$.
At first, $tail$ is set to $t_0$ (Line~\ref{algo-cas2aa}).
In this case there is no need to check if there is an abelian substring with tail length 0
 ending at position $i$ (since it has been detected in Line~\ref{algo-cas1a})
 and thus $(B[t_0], h, |\PV|-1, i-1)$ is considered as an abelian substring candidate
  (Line~\ref{algo-cas2e})
 and array $B$ is updated (Line~\ref{algo-cas2f}) since $(B[t_0],h,0,i)$ is not an abelian substring.

When $tail\ne t_0$,
let $t=\,$\textsc{getTail}$(tail,t_0,|\PV|)$.
If $\PV_{w[i-t+1\pp i]}\not\subset \PV$
 and thus $(B[tail], h, t, i-1)$ is considered as an abelian substring candidate (Line~\ref{algo-cas2e})
 and array $B$ is updated (Line~\ref{algo-cas2f}) since $(B[tail],h,t+1,i)$ is not an abelian substring.
If $\PV_{w[i-t+1\pp i]}\subset \PV$ then,
 for $tail \le k \le (t_0-1+|\PV|)\bmod |\PV|$,
 $\exists\, h'_k,t'_k$ such that $(B[k], h'_k, t'_k, i)$ is an abelian substring.
It comes directly from Prop.~\ref{prop-heap}.

At each iteration of the loop in Lines~\ref{algo-cas2b}-\ref{algo-cas2i}
 $b$ is either equal to $|w|$ or to the position of the leftmost abelian run
 ending at position $i-1$.
Thus a new candidate is found if its starting position is smaller than $b$
 (Lines~\ref{algo-cas2dbis}-\ref{algo-cas2e}).
It comes directly from Lemma~\ref{lemma-onecandidate}.

\begin{algorithm}
  \caption{\textsc{GetTail}$(tail,t_0,p)$}
  \Return{$(t_0-tail+p)\bmod p$}
\end{algorithm}

\begin{algorithm}
  \DontPrintSemicolon
  \SetAlgoNoLine
  \caption{\textsc{GetRun}$(B,tail,t_0,e,p)$}
  $b\leftarrow B[tail]$\;
  \uIf{$tail = t_0$}{
    $t\leftarrow p-1$\;
  }{
    $t\leftarrow\textsc{GetTail}(tail,t_0,p)-1$\;
  }
  $h\leftarrow(e-t-b+1)\bmod p$\;
  \Return{$(b,h,t,e)$}
\end{algorithm}

\begin{algorithm}
  \DontPrintSemicolon
  \SetAlgoNoLine
  \caption{\textsc{Runs}$(\PV,w)$}
  $j\leftarrow\textsc{Find}(\PV,w)$\;
  $(B,t_0)\leftarrow(\undefined^{|\PV|},0)$\;
  $B[t_0]\leftarrow\textsc{FindHead}(w,j-|\PV|+1,\PV)$\;
  \label{P1}
  \For{$i\leftarrow j+1$ to $|w|$}{
  \label{algo-loop1a}
    $t_0\leftarrow(t_0+1)\bmod |\PV|$\;
    \eIf{$i<|w|\textbf{ and }\PV_{w[i-|\PV|+1\pp i]}=\PV$}{
    \label{algo-cas1a}
       \uIf{$B[t_0]=\undefined$}{
      \label{algo-cas1b}
        $B[t_0]\leftarrow\textsc{FindHead}(w,i-|\PV|+1,\PV)$
        \label{algo-cas1c}
      }
    }{
     \label{algo-cas2a}
     $(b,tail)\leftarrow(|w|,t_0)$\;
     \label{algo-cas2aa}
     \Repeat{$tail = t_0$}{
     \label{algo-cas2b}
        \label{algo-cas2c}\If{$B[tail]\neq\undefined$}{
          \If{$tail=t_0\textbf{ or }i=|w|\textbf{ or }\PV_{w[i-\textsc{GetTail}(tail,t_0,|\PV|)+1\pp i]}\not\subset \PV$}{
          \label{algo-cas2d}
               \uIf{$B[tail]\leqslant b$}{
          \label{algo-cas2dbis}
            $(b,h,t,e)\leftarrow\textsc{GetRun}(B,tail,t_0,i-1,|\PV|)$
            \label{algo-cas2e}
               }
            $B[tail]\leftarrow\undefined$
            \label{algo-cas2f}
          }\lElse{
	    \label{algo-cas2g}\textbf{break}
           }
        }
      \label{algo-cas2h}$tail\leftarrow(tail+1)\bmod |\PV|$\;
      }
      \label{algo-cas2i}
      \label{algo-cas2j}\If{$\textsc{min}(B) > b \textbf{ and } e-t-h-b+1 > |\PV|$}{
        $\textsc{Output}(b,h,t,e)$
  \label{algo-loop1b}
      }
   }
}
\end{algorithm}

\subsection*{Example}

Let us see the behaviour of the algorithm on
$\Sigma = \{\sa{a},\sa{b}\}$, $w=\sa{abaababaabbb}$ and $\PV=(2,2)$:

$j=4, B = (12,12,12,12), t_0 = 0$

$B[0] = 0, B = (0, 12, 12, 12)$

$i=5$

\hspace{1cm}$t_0=1$

\hspace{1cm}$\PV_{w[2\pp 5]}\ne \PV$

\hspace{1cm}$(b,tail)=(12,1)$

\hspace{1cm}$tail = 3, 2, 1, 0$

$i=6$

\hspace{1cm}$t_0=2$

\hspace{1cm}$\PV_{w[3\pp 6]} = \PV$

\hspace{1cm}$B[2]=0, B=(0,12,0,12)$

$i=7$

\hspace{1cm}$t_0=3$

\hspace{1cm}$\PV_{w[4\pp 7]} = \PV$

\hspace{1cm}$B[3]=1, B=(0,12,0,1)$

$i=8$

\hspace{1cm}$t_0=0$

\hspace{1cm}$\PV_{w[5\pp 8]}\ne \PV$

\hspace{1cm}$(b,tail)=(12,0)$

\hspace{1cm}$(b,h,t,e) = (0,1,3,7)$

\hspace{1cm}$B[0]=12, B=(12,12,0,1)$

\hspace{1cm}$tail=1$

\hspace{1cm}$tail=2$

$i=9$

\hspace{1cm}$t_0=1$

\hspace{1cm}$\PV_{w[6\pp 9]} = \PV$

\hspace{1cm}$B[1]=3, B=(12,3,0,1)$

$i=10$

\hspace{1cm}$t_0=2$

\hspace{1cm}$\PV_{w[7\pp 10]} = \PV$

\hspace{1cm}$B[2]\ne 12$

$i=11$

\hspace{1cm}$t_0=3$

\hspace{1cm}$\PV_{w[8\pp 11]}\ne \PV$

\hspace{1cm}$(b,tail)=(12,3)$

\hspace{1cm}$(b,h,t,e) = (1,3,3,10)$

\hspace{1cm}$B[3]=12, B=(12,3,0,12)$

\hspace{1cm}$tail=0$

\hspace{1cm}$tail=1$

$i=12$

\hspace{1cm}$t_0=0$

\hspace{1cm}$i \ge 12$

\hspace{1cm}$(b,tail)=(12,0)$

\hspace{1cm}$tail=1$

\hspace{1cm}$(b,h,t,e) = (3,3,2,11)$

\hspace{1cm}$B[1]=12, B=(12,12,0,12)$

\hspace{1cm}$tail=2$

\hspace{1cm}$(b,h,t,e) = (0,3,1,11)$

\hspace{1cm}$B[1]=12, B=(12,12,12,12)$

\hspace{1cm}$tail=3$

\hspace{1cm}$tail=0$

\hspace{1cm}\textsc{Output}$((0,3,1,11)$

\subsection{Correctness and Complexity}

\begin{theorem}
The algorithm \textsc{Run}{$(\PV,w)$} computes all the abelian runs with Pa\-rikh vector $\PV$ in a string
 $w$ of length $n$ in time $O(n\times |\PV|)$ and additional space $O(\sigma+|\PV|)$.
\end{theorem}

\begin{proof}
The correctness of the algorithm comes from Corollary~\ref{coro-run},
 Lemma~\ref{lemma-onecandidate}
 and Prop.~\ref{prop-heap}.
The loop in lines~\ref{algo-loop1a}-\ref{algo-loop1b} iterates at most $n$ times. 
The loop in lines~\ref{algo-cas2b}-\ref{algo-cas2i} iterates at most $|\PV|$ times. 
The instructions in lines~\ref{algo-cas1a}, \ref{algo-cas1c} and \ref{algo-cas2d} regarding the comparison of Parikh vectors can be performed in $O(n)$ time overall, independently from the alphabet size,
 by maintaining the Parikh vector of a sliding window of length $|\PV|$ on $w$ and
 a counter $r$ of the number of differences between this Parikh vector and $\PV$.
 At each sliding step, from $w[i-|\PV|\pp i-1]$ to $w[i-|\PV|+1\pp i]$ the counters of the 
 characters $w[i-|\PV|]$ and $w[i]$ are updated, compared to their counterpart
  in $\PV$ and $r$ is updated accordingly.
The additional space comes from the Parikh vector and from the array $B$, which has $|\PV|$ elements.
\qed
\end{proof}

\section{Conclusions}

We gave an algorithm that, given a word $w$ of length $n$ and a Parikh vector $\PV$, returns all the abelian runs of period $\PV$ in $w$ in time $O(n\times |\PV|)$ and space $O(\sigma~+~|\PV|)$. The algorithm works in an online manner. To the best of our knowledge, this is the first algorithm solving the problem of searching for all the abelian runs having a given period.

We believe that further combinatorial results on the structure of the abelian runs in a word could lead to new algorithms.

One of the reviewers of this submission pointed out that our algorithm can be modified in order to achieve time complexity $O(n)$. Due to the limited time we had for preparing the final version of this paper, we did not include such improvement here. We will provide the details in a forthcoming full version of the paper. By the way, we warmly thank the reviewer for his comments.

\bibliographystyle{splncs03}         
\bibliography{main}

\end{document}